\documentclass{elsarticle}
\pdfoutput=1

\usepackage{amsmath,amssymb,amsthm,color}
\usepackage{mathdots,graphicx, float}
\usepackage[T1]{fontenc}

\newtheorem{thm}{Theorem}
\newtheorem{proposition}[thm]{Proposition}
\newtheorem{lemma}[thm]{Lemma}
\newtheorem{cor}[thm]{Corollary}
\newtheorem{definition}[thm]{Definition}

\theoremstyle{remark}
\newtheorem{remark}[thm]{Remark}

\def\C{\mathbb C}

\def\N{\mathbb N}
\def\R{\mathbb R}
\def\Rd{{{\mathbb R}^d}}

\def\Rdd{{{\mathbb R}^{2d}}}

\def\e{{\rm e}}
\def\Id{{\rm Id}}
\def\Im{{\rm Im}}
\def\Re{{\rm Re}}

\def\eps{\varepsilon}
\def\phi{\varphi}

\def\B{{\mathcal B}}

\def\K{{\mathcal K}}
\def\Y{{\mathcal Y}}
\def\X{{\mathcal X}}

\def\W{{\mathcal W}}
\def\d{{\partial}}

\def\Q{{\mathcal Q}}
\def\P{{\mathcal P}}
\def\Z{{\mathcal Z}}
\def\M{{\mathcal M}}

\def\({\left(} \def\){\right)}

   \def\lw{\left\langle} \def\rw{\right\rangle}

\date{\today}

\begin{document}

\begin{frontmatter}

\title{An invariant class of wave packets for the Wigner transform}

\author{Helge Dietert}
\address{Faculty of Mathematics, University of Cambridge, Cambridge CB3 0WA, United Kingdom, H.G.W.Dietert@maths.cam.ac.uk}

\author{Johannes Keller}
\address{Zentrum Mathematik, Technische Universit\"at M\"unchen, 80290 M\"unchen, Germany, keller@ma.tum.de}

\author{Stephanie Troppmann}
\address{Zentrum Mathematik, Technische Universit\"at M\"unchen, 80290 M\"unchen, Germany, troppman@ma.tum.de}

\begin{abstract}
  Generalised Hagedorn wave packets appear as exact solutions of Schr\"odinger equations
with quadratic, possibly complex, potential, and are given by a polynomial times a
  Gaussian. We show that the Wigner transform of
 generalised Hagedorn wave packets is a wave packet of the same type in phase space.
The proofs build on a parametrisation via Lagrangian frames and
  a detailed analysis of the polynomial prefactors, including
a novel Laguerre connection. Our findings directly imply the recently found
tensor product structure of the Wigner transform of Hagedorn wave packets.
\end{abstract}

\end{frontmatter}

\section{Introduction}\label{sec:intro}

Hagedorn wave packets are prototypes of quantum wave functions that are highly localised in both, position and momentum space.
They arrive as eigenstates of general multidimensional harmonic oscillators. The $k$-th wave packet, $k \in \N^d$,  can be written as
\begin{equation*}
\phi_k(x) = p_k(x) g(x), \quad x \in \Rd,
\end{equation*}
where $p_k$ is a polynomial of degree $|k|$ and $g$ is a Gaussian with
complex width matrix; see \cite{H80,H98,LT14}.

The favourable properties of Hagedorn wave packets have turned them into an
essential tool for various applications, notably in molecular quantum dynamics and quantum optics.
For instance, they have been used in the analysis of time-dependent
semiclassical Schr\"odinger equations \cite{Ha85,HJ00,LST15},
the construction of numerical methods for high frequency problems \cite{Lu08,FGL09},
and the representation of quasiprobabilities in quantum optics~\cite{W99}.

In the literature, Hagedorn wave packets also coexist as generalised squeezed states; see \cite{CR12}.
These states are generated by applying translation and squeezing operators to the eigenstates of
the standard harmonic oscillator. 
However, this construction and Hagedorn's approach produce, up to a phase factor, the same wave packets, see~\cite{LT14}.

In this paper our aim is to analyse the representation of Hagedorn wave packets on phase space, that is, their Wigner functions, and give
a comprehensive description of the appearing structures.

So far, Wigner functions have been classified for Hermite functions, see \cite{F89, Th93}, by using the well-known
Hermite-Laguerre connection. This link has been generalised in \cite{LT14}, showing that the Wigner functions
of a Hagedorn wave packet factorises regardless of the structure of the wave packet itself. Another approach using ladder operators
in the two-dimensional setting can be found in \cite{DM94}.

Building on the generating function of the wave packets from \cite{Hag15}, we are able to present a general formula
for the Wigner functions of Hagedorn wave packets.
As our main result, in Theorem~\ref{thm:wigner-hagedorn} we show that the wave packets in phase space are given by Hagedorn wave packets of doubled dimension.
This insight allows to lift all qualities of the wave packets to phase space and, thus, enables new perspectives and approaches
in modelling and approximation theory. We find that the structure of the Wigner functions can be explained with the
 properties of the wave packets.
A connection between the polynomial prefactor in the Wigner functions and Laguerre polynomials then reproduces the
previously found factorisation results.


Wigner functions play a central role in microlocal analysis and semiclassical quantum mechanics; see, e.g., \cite{dG11,Z12}.
In particular, various important semiclassical approximations and algorithms for the simulation of
quantum molecular dynamics combine the Wigner function of the initial state with classical dynamics on
phase space, see, e.g., \cite{BR02,LR10,GL14}. Our results can contribute to improving these methods by providing new
methods for the approximation of Wigner functions of more general states, and extending the applicability of wave packet propagation methods to phase space.


\subsection{Outline}
Motivated by the form
the wave packets attain under non-unitary time evolution, in Chapter 2 we start with a generalised definition for
Hagedorn wave packets; see also \cite{LST15}. For our analysis we follow a geometric approach based on
Lagrangian frames. However, all conditions can easily be translated into the classical ones found in \cite{{Lu08}} or \cite{H98}.

In order to find criterions for tensor factorisation, in Chapter 3 we analyse the polynomial prefactors $p_k$ of the wave packets.
We characterise these polynomials by a three-term
recursion relation, a generating function as well as a ladder operator. All definitions are closely related to the
corresponding formulas for Hermite polynomials, and thus suggest an interpretation as multivariate Hermite polynomials.
However, the polynomials are not simple tensor products of univariate Hermite polynomials, but generically exhibit a more
complex structure, as shown in Proposition~\ref{lem:laguerre+reduction}. Nevertheless, our findings can be read as reasonable multivariate
extension of the relations between Hermite and Laguerre polynomials discussed in \cite{Th93}.

In Chapter 4, we use the generating function to identify the Wigner function of Hagedorn wave packets with a wave packet of doubled
dimension. This result combined with our polynomial analysis from Chapter 3 give a detailed picture of Hagedorn wave packets
in phase space that constitutes the core of our manuscript.

Finally, in Chapter 5, we illustrate characteristic examples for both, polynomials and wave packets in two dimensions.


\section{Generalised Hagedorn wave packets}\label{sec:Hag}

In this chapter we adopt the viewpoint of~\cite{LST15}, and briefly sketch the geometric approach to defining Hagedorn wave packets.
Moreover, we extend this framework by decoupling the raising operator from the ground state.
This type of generalised wave packets appears for example in the context of non-selfadjoint evolution problems, see \cite{H95,GS12,LST15}.
\subsection{Lagrangian frames and ground states}\label{sec:GrState}

We consider the classical phase space $T^*\Rd = \Rdd$, equipped with the standard symplectic form
\begin{equation*}
 \Omega = \left(\begin{array}{cc} 0 & -\Id_d \\ \Id_d & 0 \end{array}\right) \in \R^{2d\times 2d}.
\end{equation*}
\begin{definition}
A matrix $Z\in \C^{2d\times d}$ is called a \emph{Lagrangian frame}, if
\begin{equation}\label{eq:isotropy}
Z^T\Omega Z = 0
\end{equation}
and ${\rm rank} \ Z = d$. Furthermore, $Z$ is called \emph{normalised}, if
\begin{equation}\label{eq:normalisation}
\frac{i}2 Z^* \Omega Z = \Id_{d}.
\end{equation}
\end{definition}

In the following, we will denote typical phase space points by $z = (q,p) \in \Rdd$, with $q\in \Rd$ being the position, and
$p\in \Rd$ the momentum variable. Similarly, we write $Z = (Q;P)$ with $Q,P\in \C^{d\times d}$ for Lagrangian frames.
 With this denotation, the isotropy condition \eqref{eq:isotropy} and the normalisation condition \eqref{eq:normalisation}
 coincide with the symplecticity condition in \cite{{Lu08}} or, by taking $Q = A$ and $P= i B$, with
 Hagedorn's original definition in \cite{H98}.

The image of a normalised Lagrangian frame $Z$ is a complex Lagrangian subspace $\rm{range} \ Z \subset \C^{2d}$.
All Lagrangian frames spanning the same subspace $L$ are related by unitary transformations, see, e.g., \cite[\S2.1]{LST15}, and thus
have the same Hermitian square $Z Z^*$. Hence, considering the matrix
\begin{equation}\label{eq:sympl_metric}
G_{Z}=  \Omega^T \Re(ZZ^*) \Omega =\begin{pmatrix} PP^* & -PQ^* + i\Id_d\\ -QP^* - i\Id_d & QQ^*\end{pmatrix}
\end{equation}
is a convenient way to classify the Lagrangian subspace $\rm{range} \ Z$.
The matrix $G_{Z}$ is real, symmetric, positive definite and symplectic, and thus called the symplectic metric associated with the
Lagrangian subspace $\rm{range} \ Z$. The matrix $G_Z$ plays an important role when lifting the wave packets to phase space.

Following the idea of Hagedorn, we can associate a coherent ground state with any normalised Lagrangian frame.
\begin{lemma}[Ground state]
Let $Z=(Q;P)$ be a normalised Lagrangian frame. Then, $Q$ and $P$ are invertible and
$$
\Im\(PQ^{-1}\) = (QQ^{*})^{-1} > 0.
$$
In particular, for $\eps>0$,
\begin{equation}\label{eq:hagedorn_ground}
\phi^{\eps}_0[Z](x) = (\pi \eps)^{-\tfrac{d}{4}} \det(Q)^{-\tfrac{1}{2}} \exp(\tfrac{i}{2\eps} x^T PQ^{-1}x)
\end{equation}
is a square integrable function with $\|\phi^{\eps}_0[Z]\|_{L^2} = 1$.
\end{lemma}
A rigorous proof of this result can for example be found in \cite{Lu08}. Briefly, isotropy \eqref{eq:isotropy} ensures $\phi^{\eps}_0[Z] \in L^2(\R^d)$, while $\phi^{\eps}_0[Z]$ is normalised if and only if $Z$ is normalised.
\subsection{Excited states and spectral properties}\label{sec:ExStates}

Let $\widehat{q}$ denote the position operator and $\widehat{p}$ the momentum operator.
Analogously to  phase space points, we write $\widehat{z} = \begin{pmatrix} \widehat{q} , \widehat{p}\end{pmatrix}$.
Using this notation, we can define a linear operator
\begin{equation}\label{eq:ladder_pos}
A^{\dagger}[Y]= \tfrac{i}{\sqrt{2 \eps}} Y^* \Omega \widehat{z} 
\end{equation}
associated with a normalised Lagrangian frame $Y$.
The components of $A^{\dagger}[Y]$ commute due to the isotropy of $Y$.
Starting from a ground state $\phi^\eps_0[Z]$, the generalised Hagedorn wave packets $\phi^\eps_k[Z,Y]$ then are constructed
via $A^{\dagger}[Y]$ as follows.
\begin{definition}[Generalised wave packets]
Let $k \in \N^d$ and $Z,Y\in \C^{2d\times d}$ be normalised Lagrangian frames. Then, the $k$-th Hagedorn wave packet is defined as
\begin{equation}\label{eq:hagedorn_raising}
\phi^{\eps}_k[Z,Y] = \tfrac{1}{\sqrt{k!}} (A^{\dagger}[Y])^k \phi^{\eps}_0[Z]
\end{equation}
with standard multiindex notation, $(A^{\dagger}[Y])^k = A^{\dagger}[Y]_1^{k_1}\ldots A^{\dagger}[Y]_d^{k_d}$.
\end{definition}
Based on this definition we will refer to $A^{\dagger}[Y]$ as raising operator.

If $Y=Z$, the above definition yields the standard Hagedorn wave packets $\{\phi_k^{\eps}[Z]\}_{k \in \N^d}$ that
form an orthonormal basis of $L^2(\Rd)$. In this case, the adjoint operator of  $A^{\dagger}[Z]$,
\[
A[Z]= -\tfrac{i}{\sqrt{2 \eps}} Y^T \Omega \widehat{z}, \quad A_j[Z] \phi^{\eps}_{k}[Z] = \sqrt{k_j}\phi^{\eps}_{k-e_j}[Z],
\]
plays the role of a lowering operator.
This property does not hold for the generalised wave packets.

The construction of Hagedorn wave packets by means of $A^\dagger[Y]$ implies
\begin{equation*}
\phi_k^{\eps}[Z,Y] = p_k^\eps[Z,Y] \phi_0^{\eps}[Z],
\end{equation*}
where $p_k^\eps[Z,Y]$, $k\in \N^d$, is a polynomial of total degree $|k|$.
\begin{proposition}\label{prop:polynomial_hagedorn_trafo}
  Let $Z = (Q;P)\in \C^{2d \times d} $ and $Y= (X;K) \in \C^{2d \times d} $ be two normalised Lagrangian frames and $B=\tfrac i2 Z^*\Omega Y$. Then, for $k\in \N^d$ it holds
  \begin{equation}
    \label{eq:polynomial_hagedorn_trafo}
    \phi_k^{\eps}[Z,Y](x) = \frac1{\sqrt{2^{|k|}k!}}
    q_k^M \( \tfrac{1}{\sqrt{\eps}} B^* Q^{-1}x \)  \phi_0^{\eps}[Z](x)
  \end{equation}
  with $M = \tfrac{1}{4} Y^* G_{Z} \overline{Y} + B^* Q^{-1} \overline{QB}$  and
the polynomials $\{q^M_k\}_{k\in \N^d}$ that are recursively defined by $q^M_0 \equiv 1$ and
  \begin{equation*}
  (q^M_{k+e_j} (x) )_{j=1}^d = 2 x q^M_k(x) - 2 M\cdot (k_jq^M_{k-e_j}(x))_{j=1}^d.
  \end{equation*}
  In the special case $Y=Z$ we find $B= \Id$ and $M= Q^{-1}\overline{Q}$.
\end{proposition}
Proposition~\ref{prop:polynomial_hagedorn_trafo} is easily
  proven by identifying the found
  raising operator with the recursive definition of $\{\phi^\eps_k\}_k$, see
  Appendix \ref{app:Proof}. A similar result for special choices
of $Y$ can be found in \cite{LST15}.
Proposition~\ref{prop:polynomial_hagedorn_trafo}  reveals that
the structure of the Hagedorn wave packets can be understood
  by studying the polynomial prefactors $\{q^M_k\}_k$.  This polynomial
  analysis is conducted in the next section.

In this chapter we only introduced wave packets centered at
  the origin.
 In general, one  obtains wave packets centered at any phase space point
  $z=(q,p) \in \R^{2d}$
  by applying the linear Heisenberg-Weyl operator
\begin{equation*}
(T_z \psi)(x) = \e^{ip^T(x-q/2)/\eps} \psi(x-q)
\end{equation*}
to the corresponding wave packet centered at the origin.

Also in the remaining parts of this paper, we  only
present our analysis for the wave packets
centered at the origin. However, we stress that
all of our results and proofs easily adapt to the more general case.

%
%

\section{Analysis of the polynomial prefactor}

In this chapter we analyse the polynomial prefactors of Hagedorn's wave packets, which are
defined for a symmetric and unitary matrix $M \in \C^{d\times d}$ via the three-term
recursion relation (TTRR)
\begin{equation}
\label{eq:general_TRR}
( q^M_{k+e_j} (x) )_{j=1}^d = 2 x q^M_k(x) - 2 M\cdot (k_jq^M_{k-e_j}(x))_{j=1}^d,
\end{equation}
see Proposition~\ref{prop:polynomial_hagedorn_trafo},
where $e_j$ denotes the $j$-th unit vector in $\Rd$ with boundary
conditions $q^M_0\equiv 1$ and $q^M_\ell \equiv 0$ for all $\ell\notin \N^d$. This
recursion is well-defined as $M$ is symmetric. Equivalently, the
polynomials $q^M_k$ could be defined via their generating function or
raising operator, both of which we derive in this section.

Surprisingly, this class of polynomials seems to be  little studied in the literature
  found by the authors. A noteable exception is \cite{E39} who defined them
  through their generating function and showed a generalised Mehler formula.
 In \cite{DM94} one can find results for the two-dimensional case.

\subsection{Generating function and raising operator}
To get an intuition, we quickly discuss the univariate polynomials first.
In one dimension \eqref{eq:general_TRR} simplifies to
$$
H^{\lambda}_{n+1}(x) = 2x H^{\lambda}_n(x)-2 \lambda n H^{\lambda}_{n-1}(x)
$$
with $\lambda \in \R$. Starting from $H^{\lambda}_0 \equiv 1$
and $H^{\lambda}_{-1} \equiv 0$, for $\lambda = 1$ we produce the usual Hermite polynomials. For $\lambda \neq 0$ we find rescaled Hermite polynomials
$$
H^{\lambda}_n = \lambda^{\tfrac{n}{2}} H^1_n\(\frac{x}{\sqrt{\lambda}}\).
$$
To complete the picture, for $\lambda = 0$, we generate monomials $H^0_n = (2x)^n$.

In the multivariate case the polynomials may emerge as simple tensor products of Hermite
polynomials, but typically attain a more involved structure. Recently,
in \cite{Hag15}, there has been published a formula for the generating
function of (rescaled) polynomials of type \eqref{eq:general_TRR}.
\begin{lemma}[Generating function] \label{lem:gen_func}
Let $M\in \C^{d \times d}$ be symmetric and unitary. Then, the generating
function of the polynomials $\{q^M_k\}_{k \in \N^d}$ is given by
\begin{equation}\label{eq:genfct}
f(x,t) =\sum_{k \in \N^d} \frac{t^k}{k!} q_k(x) = \exp(2 x^T t - t ^T Mt).
\end{equation}
\end{lemma}
To give a more profound characterisation of the polynomials,
we also look at the raising operator and the gradient
formula for the polynomials $q^M_k$.  In the language of Dirac
ladders, the gradient plays the role of a lowering operator or
annihilator for the polynomials.

\begin{lemma}[Ladder operators]\label{lem:raising_operator}
Let $b_M^\dagger = 2x - M\nabla_x$ and $k\in \N^d$. Then, it holds
\begin{equation}
(q^M_{k + e_j})^d_{j=1} = b_M^\dagger q^M_{k}~\text{\quad\quad~and~
\quad\quad~}\nabla q^M_k = 2 (k_j q^M_{k-e_j})^d_{j=1}.
\end{equation}
\end{lemma}

\begin{proof}
  The generating function satisfies $\nabla_x f = 2t f$ and
  \begin{equation*}
    \nabla_t f = 2x f - M \nabla_x f =b_M^\dagger f.
  \end{equation*}
  By the definition of the generating function $f(x,t) $ 
  it follows that
   $$
  \partial_{t_j}f(x,t) = \sum_k  \frac{t^{k-e_j}}{(k-e_j)!} q^M_k(x) = \sum_k  \frac{t^{k}}{k!} q^M_{k+e_j}(x),
  $$
  i.e. $\partial_{t_j}f$ simply shifts the summation index by one into the direction of $e_j$.
  The gradient formula then follows immediately from
  (\ref{eq:general_TRR}).
\end{proof}

As a consequence of~\eqref{eq:general_TRR} and Lemma~\ref{lem:raising_operator},
the polynomials satisfy the Rodrigues formula
\begin{eqnarray}
\label{eq:rod}
q^M_k(x) =\exp\left(x^T M^{-1} x\right)(-M \nabla )^k \exp\left(-x^T M^{-1} x\right).
\end{eqnarray}


\subsection{Factorisation and eigenvectors}
The structure of the multivariate polynomials $q^M_k$ crucially
depends on the matrix $M$ in the TTRR~\eqref{eq:general_TRR}. In fact, the polynomials $q^M_k$ factorise for all $k\in\N^d$ if and
only if their generating function factorises. We summarize this observation in the following remark.


\begin{remark}\label{lem:factorisation}
By inspecting Lemma~\ref{lem:gen_func} one observes that the generating function
of the polynomials $\{q^M_k\}_{k \in \N^d}$
factorises into $m$ lower-dimensional generating functions of the
same type if and only if there is a relabeling of the coordinates such
that $M$ is block diagonal with $m$ blocks. In this case,
the polynomials $q^M_k$ are tensor products of $m$ lower-variate
polynomials of the same type for all $k\in\N^d$.

In particular, the polynomials $q^M_k$ are tensor products of
univariate polynomials for all $k$ if $M$ is diagonal.  On the contrary, if $M$ has nonzero offdiagonal entries,  the polynomials
$q_k^M$ are not simple tensor products anymore.
\end{remark}


The polynomials $q^M_k$ form a set of simultaneous eigenvectors for the operators
\[
T_j = \frac{(b_M^\dagger)_j \d_{x_j} +  \d_{x_j}(b_M^\dagger)_j }2 = (1+2x_j\d_{x_j}) - \d_{x_j} (M\nabla)_j~,~j=1,\hdots,d,
\]
acting, for example, on the space of polynomials on $\Rd$. This can easily be seen from
\begin{align*}
T_j q^M_k &=\tfrac12\( (b_M^\dagger)_j 2k_j q^M_{k-e_j} +  \d_{x_j} q^M_{k+e_j} \) \\
& = (2k_j +1)q^M_k,
\end{align*}
similarly as in the computation of harmonic oscillator eigenvalues.
Hence, $q^M_k$ is an eigenvector belonging to the eigenvalue $2k_j +1$ of
$T_j$.  Moreover, for $j\neq k$, the operators $T_j$ and $T_k$ commute
if and only if $M_{jk} = M_{kj} = 0$. In other words, by the remark above, commutation of the operators
reflects the factorisation of the polynomials.

\subsection{Laguerre connection}

As noted in Remark~\ref{lem:factorisation}, the polynomials~$q_k^M$
factorise for all $k$ if and only if $M$ is block-diagonal.
Therefore, our aim is to express the general polynomials as a
linear combination of tensor products by deleting off-diagonal entries
of $M$.

It turns out, that $q_k^M$ can be rewritten as a Laguerre polynomial
\begin{equation*}
L_n^{(\alpha)}(x) = \sum^n_{j=0} \left(\begin{array}{c}n+\alpha \\
n-j \end{array}\right) \frac{1}{j!}(-x)^j~,~~n\in \N~,~~\alpha\geq 0,
\end{equation*}
of the raising operators corresponding to the polynomials generated by
a reduced matrix applied to $1$.
This generalised Laguerre connection adds a new point to the long list
of relations between Hermite and Laguerre polynomials, see, e.g.,~\cite{Th93}.
More precisely, we express $q_k^M$ via raising operators
$b^\dagger_{M[n,m]}$, where $M[n,m]$ denotes the matrix $M$ with
deleted offdiagonal entries $M_{nm}$ and $M_{mn}$, i.e.
\begin{equation*}
  \(M[n,m] \)_{ij}=
  \begin{cases}
    0~, & \{i,j\} = \{n , m\}, \\
    M_{ij}~, & \text{otherwise}.
  \end{cases}
\end{equation*}

\begin{proposition}[Laguerre connection]\label{lem:laguerre+reduction}
Let $M\in \C^{d\times d}$ be symmetric, and
$M_{nm} = \lambda \neq 0$ for some $n \neq m$. Suppose $k\in \N^d$ with
$k_n \geq k_m$. Then,
\begin{align}\label{eq:laguerre_eliminate}
q_k^M(x) &= \(b_M^\dagger\)^k 1  = \nonumber\(c^\dagger\)^{k - k_m(e_n+e_m)}
(-2\lambda)^{k_m} k_m!
L^{(k_n-k_m)}_{k_m}\(\tfrac{1}{2\lambda} c_n^\dagger c_m^\dagger \) 1,
\end{align}
where $c^\dagger=b_{M[n,m]}^\dagger$
denotes the polynomial raising operator for the reduced matrix $M[n,m]$.
The case $k_n < k_m$ is analogous.
\end{proposition}

\begin{proof}
Let $f_M$ denote the generating function of the polynomials $q_k^M$ derived in
Lemma~\ref{lem:gen_func}, and define $k[n,m] = k-e_nk_n - e_m k_m$. Then, one computes the series expansion in $t$ as
 \begin{align*}
    f_M(x,t) =& f_{M[n,m]}(x,t) \exp(-2\lambda t_nt_m) \\
    =& \(\sum_{k\in \N^d} \frac{t^k}{k!} \(c^\dagger\)^k \)\cdot
    \( 1- 2\lambda t_nt_m + \tfrac{1}{2!}(2\lambda t_nt_m)^2 - \hdots   \) \\
    =& \sum_{k\in \N^d} \frac{t^k \(c^\dagger\)^{k[n,m] } }{(k[n,m])!}
     \(
     \sum^{\min(k_n,k_m)}_{j=0}  \frac{(-2 \lambda)^j}{j!}
     \cdot  \frac{(  c_n^\dagger)^{k_n-j} }{(k_n-j)!}\cdot
     \frac{( c_m^\dagger)^{k_m-j}}{(k_m-j)!}
      \ 1
      \).
      \end{align*}
  Due to the definition of the generating function, this implies
  \begin{equation}\label{eq:ladder_matrix_delete}
 q_k^M(x) = \(c^\dagger\)^{k[n,m]} \sum_{j=0}^{\min(k_n,k_m)}
 \frac{k_n!k_m! (-2 \lambda)^j }{j!(k_n-j)!(k_m-j)!}  (c_n^\dagger)^{k_n-j}( c_m^\dagger)^{k_m-j} \ 1
  \end{equation}
and we can reorder the sum by means of the index $\ell= k_m-j\geq 0$, since
$k_n \geq k_m$ holds by assumption. This finally leads to
\begin{align*}
 q_k^M(x)
 &= \(c^\dagger\)^{k[n,m]} (-2 \lambda)^{k_m}  k_m!  (c_n^\dagger)^{k_n-k_m}
 \sum_{\ell=0}^{k_m} \frac{k_n!  (- \tfrac{1}{2 \lambda}
 c_n^\dagger c_m^\dagger)^{k_m-\ell}}{(k_m-\ell)!(k_n-k_m+\ell)! \ell!}  \ 1 \\
 &= \(c^\dagger\)^{k[n,m]}  (-2 \lambda)^{k_m} k_m!  (c_n^\dagger)^{k_n-k_m}
 L^{(k_n-k_m)}_{k_m}\(\tfrac{1}{2\lambda} c_n^\dagger c_m^\dagger \)  \ 1,
\end{align*}
where we utilised that $c_n^\dagger$ and $c_m^\dagger$ commute.
\end{proof}


A direct, illustrative consequence can be deduced for the two-dimensional case, for which
$$
M = \left(\begin{array}{cc} \lambda_1 & \lambda_3 \\  \lambda_3 &
\lambda_2 \end{array}\right) ~\quad\text{with~}\lambda_1,\lambda_2,\lambda_3 \in \C.
$$

\begin{cor}
Let $k\in \N^2$ with $k_1\geq k_2$. Then,
\begin{equation}\label{eq:2d+complete+form}
q_k^M(x) = \begin{cases}
(-2\lambda_3)^{k_2} k_2! (a^\dagger_{\lambda_1})^{k_1-k_2}
L^{(k_1-k_2)}_{k_2}\(\tfrac{1}{2\lambda_3}a^\dagger_{\lambda_1}a^\dagger_{\lambda_2}\) 1,
& \lambda_3 \neq 0,\\
\(a^\dagger_{\lambda_1}\)^{k_1}  \(a^\dagger_{\lambda_2}\)^{k_2}   1,& \lambda_3 = 0.
\end{cases}
\end{equation}
The case $k_1\leq  k_2$ is analogous.
\end{cor}


For the case $\lambda_3 = 0$, Equation (\ref{eq:2d+complete+form})
directly carries out to the factorisation
\begin{equation*}
q_k^M(x) = H^{\lambda_1}_{k_1}(x_1)H^{\lambda_2}_{k_2}(x_2).
\end{equation*}
In the case $\lambda_3 \neq 0$, (\ref{eq:2d+complete+form}) guarantees that each
$q^M_k$ is just a linear combination of at most $\min\{k_1,k_2\}$ many tensor products
of the form
\begin{equation}\label{eq:lambda2ist0}
\(a^\dagger_{\lambda_1}\)^n \(a^\dagger_{\lambda_2}\)^m1 = H^{\lambda_1}_{n}(x_1)H^{\lambda_2}_{m}(x_2)
\end{equation}
where $n-m = k_1-k_2$ and  $k_1-k_2 \leq n \leq k_1$, $m\leq   k_2$.
Moreover, if $\lambda_1 = \lambda_2 = 0$, the diagonal creation operators $a_0^{\dagger}$
  produce monomials, and we obtain the formula
  \begin{equation}\label{eq:2d_offdiag_laguerre}
    q_k^{M}(x)= (- \lambda_3)^{k_2} k_2! 2^{k_1} x_1^{k_1-k_2}L^{(k_1-k_2)}_{k_2}(\tfrac{2}{\lambda_3}x_1x_2)
  \end{equation}
  whenever $\lambda_3 \neq 0$ and $k_1 \geq k_2$. See
  also~\S\ref{sec:examples_pol} for illustrations.

 We note that by applying Proposition~\ref{lem:laguerre+reduction} iteratively one obtains an expansion of the general polynomials $q_k^M$ in terms
 of tensor product Hermite polynomials, see Appendix~\ref{app:Tensor}.

\section{Hagedorn wave packets in phase space}\label{sec:PhaseSpace}

There are various representations of quantum systems on the classical
phase space $\Rdd$. The most popular one is the Weyl correspondence,
where suitable selfadjoint quantum observables
$A:L^2(\Rd) \to L^2(\Rd)$ are represented by their semiclassical Weyl
symbol $a:\Rdd \to \R$, and wave functions $\phi,\psi \in L^2(\Rd)$ by
their Wigner function
\begin{equation}
  \label{eq:Wigner}
  \W^{\eps}(\phi, \psi)(x, \xi) = (2 \pi \eps)^{-d} \int_{\R^d} \overline{\phi}(x+\tfrac{y}{2}) \psi(x-\tfrac{y}{2}) \e^{i\xi^T y/\eps} \ dy.
\end{equation}
Then, matrix elements of $A$ can be computed via the phase space integral
\begin{equation}\label{eq:Weyl_expectation}
  \lw \phi, A \psi \rw_{L^2(\Rd)} = \int_\Rdd \W^{\eps}(\phi, \psi)(z) a(z) dz.
\end{equation}

In this chapter our aim is to analyse the Wigner functions of Hagedorn wave packets. For readability, we write
\begin{equation*}
  \W^\eps_{k,\ell}[Z,Y] =  \W^\eps(  \phi_k^\eps[Z,Y],   \phi_\ell^\eps[Z,Y])
\end{equation*}
for $k,\ell \in\N^d$, and regard $(k,\ell)$ as a multiindex in
$\N^{2d}$. Note, that by invoking~\cite[Equation (9.25)]{dG11} we
could also allow two wave packets that have different phase space
centers.


As our main result, we prove that the Wigner functions of Hagedorn wave packets on
$\Rd$ are given by related Hagedorn wave packets on the phase space $\Rdd$.
 For proving this invariance result, we first introduce a phase space lift of
Lagrangian frames.

\begin{lemma}\label{lem:phase+lift}
  Let $Z = (Q;P) \in \C^{2d\times d}$ and $Y = (X;K) \in \C^{2d\times d}$ be two normalised Lagrangian
  frames and $B = \tfrac{i}{2} Z^* \Omega Y$. We define the lifted matrices as
  \begin{align*}
    \Z  = \begin{pmatrix} \Q \\ \P \end{pmatrix}= \begin{pmatrix} \tfrac12  \overline{Z} & \tfrac12Z \\ \Omega \overline{Z} & -\Omega Z \end{pmatrix} \quad \text{and} \quad  \Y = \begin{pmatrix} \X \\ \K \end{pmatrix}= \begin{pmatrix} \tfrac12 \overline{Y} & \tfrac12 Y \\ \Omega \overline{Y} & - \Omega Y \end{pmatrix}  \end{align*}
  in $\C^{4d\times 2d}$ with the symplectic form $\Omega_{4d} = \Omega \otimes \Id_2$ on
    $\R^{4d}$. These have the following properties
  \begin{enumerate}
  \item $\Z $ and $\Y$ are normalised Lagrangian frames if and only if $Z$ and $Y$ are.
  \item The symplectic metric fulfills ~$\P \Q^{-1} = 2iG_{Z}$.
  \item For the lifted frames we have
    \begin{equation*}
      \B = \tfrac{i}{2} \Z^* \Omega_{4d} \Y
      = \begin{pmatrix} \overline{B}  &0 \\ 0 & B \end{pmatrix}.
    \end{equation*}
  \item\label{eq:factor_M}
    The lifted recursion matrix satisfies
    \begin{align*}
      \M = \begin{pmatrix}
            -\tfrac{1}{4} Y^T G_{Z} Y & (B^*B)^T \\
             B^*B &\tfrac{1}{4} \overline{Y^T G_{Z} Y}
           \end{pmatrix}.
    \end{align*}
  \item \label{eq:factor_M_special} For the special case $Y = Z$
    it holds $\B= \Id_{2d}$ and
    \begin{equation*}
      \M =  \Q^{-1}\overline{\Q} = \begin{pmatrix}0 & \Id_d\\ \Id_d & 0 \end{pmatrix}.
    \end{equation*}
  \end{enumerate}
\end{lemma}

\begin{proof}
For the first assertion, one computes
\begin{align*}
\Y^T \Omega_{4d} \Y&= \begin{pmatrix} -\overline{Y^T \Omega Y} & 0 \\ 0 & Y^T \Omega Y \end{pmatrix}, \quad \Y^* \Omega_{4d} \Y= \begin{pmatrix} -\overline{Y^* \Omega Y} & 0 \\ 0 & Y^* \Omega Y \end{pmatrix} ,
\end{align*}
and notes that $ \Y$ is normalised if and only if $Y$ is.

One can easily prove that $\Q^{-1} = i \P^*$. Hence, the second part
follows from
\begin{align*}
  \P\P^* =\left( \Omega \overline{Z} \ -\Omega Z\right)
  \begin{pmatrix} -Z^T\Omega \\ Z^* \Omega \end{pmatrix}
  = -\Omega \left( \overline{Z}Z^T + ZZ^* \right) \Omega = 2 G_Z .
\end{align*}
The formula for $\B$ is a direct computation.
For the claimed form of $\M$ first compute
\begin{align*}
  \Q^{-1}\overline{\Q}
  = \tfrac{i}{2}
  \begin{pmatrix} -Z^T\Omega Z&  -\overline{Z^* \Omega
      Z}\\ Z^*\Omega Z  & \overline{Z^T\Omega Z}
  \end{pmatrix}
  = \begin{pmatrix} 0 & \Id_d \\   \Id_d & 0 \end{pmatrix}.
\end{align*}
By calculating
\begin{align*}
\Y^T G_{\Z} \Y = \Y^T \Omega^T_{4d} \Z \Z^* \Omega_{4d} \Y = \begin{pmatrix} - \overline{Y^T G_{Z} Y} &0\\ 0 & Y^T G_{Z} Y\end{pmatrix}
\end{align*}
the claim follows.
The last case follows from the isotropy condition under the additional
assumption.
\end{proof}
Since $\Z$ is again a Lagrangian frame, we can lift all our previous results
to the phase space and consequentially find a family of Hagedorn wave
packets in doubled dimension.
In order to avoid confusion, we denote Hagedorn wave packets on
phase space by upper case letters~$\Phi^\eps_{(k,\ell)}$.
In particular, a direct computation shows that
\begin{equation}
\W^\eps(\phi_0[Z])(z) = (\pi \eps)^{-d} \e^{-z^T G_Z z/\eps} = (2 \pi \eps)^{-d/2} \Phi^\eps_{(0,0)}[\Z](z),
\end{equation}
see, e.g.,~\cite{LT14}. For excited wave packets the following result holds true.

\begin{thm}\label{thm:wigner-hagedorn}
  Assume that $Z,Y \in \C^{2d\times d}$ are normalised
  Lagrangian frames. Then, for $k,\ell\in \N^d$, the Wigner function
  $\W_{k,\ell}[Z,Y]$ is a Hagedorn wave packet on
  phase space,
  \begin{equation*}
    \W^\eps_{k,\ell}[Z,Y]  = (2\pi \eps)^{-d/2} \Phi_{(k,\ell)}^\eps[\Z,\Y].
  \end{equation*}
  Consequently, it holds
\begin{equation*}
\W^\eps_{k,\ell}[Z,Y](z)  =  \frac{(2\pi \eps)^{-d/2}  }{\sqrt{2^{|k| + |\ell|}k!\ell!}}
 q_{(k,\ell)}^\M \( \tfrac{1}{\sqrt{\eps}} \B^* \Q^{-1}z \)  \Phi_0^{\eps}[\Z](z),
\end{equation*}
where the lifted matrices $\Y,\Z,\Q,\P,\B$, and $\M$
have been defined in Lemma~\ref{lem:phase+lift}.
\end{thm}

\begin{proof}
  For the generalised Hagedorn wave packets, we find the generating
  function
  \begin{equation}\label{eq:gen_fun_hagedorn_pol}
    \sum_{k\in\N^d} \frac{t^k}{\sqrt{k!}} \sqrt{2^{|k|}} \phi_k^\eps[Z,Y](x)
    = \e^{\frac{2}{\sqrt{\eps}} t^T B^* Q^{-1}x - t^T M t}
    \phi_0^{\eps}[Z](x) =: h_t(x)
  \end{equation}
  by identifying the polynomial factors in Proposition~\ref{prop:polynomial_hagedorn_trafo}.
  Note that $h_t$ is
  a Gaussian function in $x$ and $t$. By writing $z=(x,\xi)$
  and $v=(t,s)$, we can easily compute the Wigner transformation of $h_t$ and $h_s$ as
  \begin{align*}
 \W^\eps(h_t,h_s)(z) &=(2 \pi \eps)^{-d} \int_{\R^d}
      \overline{h_t(x+\frac y2)}
      h_s(x-\frac y2)  \e^{i\xi^T y/\eps} \ dy \\
      &= (\pi \eps)^{-d}  \e^{ - z^T G_{Z}z/ \eps}
      \e^{ \frac{2}{\sqrt{\eps}} v^T\B^* \Q^{-1}z + v^T\M v} ,
  \end{align*}
  where we identified the lifted matrices from  Lemma~\ref{lem:phase+lift}.

  Then, by formally interchanging the order of integration and summation, it is
  clear that
  \begin{equation*}
     \W^\eps(h_t,h_s)(z)
    = \sum_{k,\ell \in\N^d} \frac{t^ks^\ell}{\sqrt{k!\ell!}}
    \sqrt{2^{|k| + |\ell|}} \W_{k,\ell}[Z,Y](z),
  \end{equation*}
  which implies the result. For a rigorous justification, note that
  due to the Gaussian decay we can differentiate under the integral
  sign with respect to $v=(t,s)$. Evaluating the differentiated
  function at $v=(t,s)=(0,0)$ then shows the result.
\end{proof}


\begin{remark}
By invoking part (5) of Lemma \ref{lem:phase+lift}, Theorem~\ref{thm:wigner-hagedorn} explains the factorisation result for
Hagedorn wave packets that has recently been discovered in \cite{LT14}.
Namely, the polynomial prefactor of $\W^\eps_{k,\ell}[Z,Z]$ is a product of $d$ polynomials,
\begin{align*}
  \W^\eps_{k,\ell}[Z,Z](z)
  &=  \frac{(2\pi \eps)^{-d/2}  }{\sqrt{2^{|k| + |\ell|}k!\ell!}}
    \Phi_0^\eps[\Z](z)
\prod_{j=1}^d q_{(k_j,\ell_j)}^{N} \( \(\tfrac{1}{\sqrt{\eps}}\Q^{-1}z\)_j, \(\tfrac{1}{\sqrt{\eps}}\Q^{-1}z\)_{d+j} \) ,
\end{align*}
where, for $k_j\geq \ell_j$ and $N=\begin{pmatrix} 0&1\\1 &0\end{pmatrix}$,
\[
 q_{(k_j,\ell_j)}^{N}(x_1,x_2) = (- 1)^{\ell_j} \ell_j! 2^{k_j} x_1^{k_j-\ell_j}L^{(k_j-\ell_j)}_{\ell_j}(2x_1x_2)
\]
 is a Laguerre  polynomial of the form~\eqref{eq:2d_offdiag_laguerre}.
\end{remark}

\section{Examples}\label{sec:Examples}

\subsection{Polynomial prefactor}\label{sec:examples_pol}

In order to illustrate different types of polynomials $q_k^M$,
we present various examples for the nodal sets of two-dimensional polynomials $q^M_{i,j}$, for
$(i,j) \in \N^2$. For simplicity, we restrict ourselves to real matrices $M$
such that
the polynomials generated by the TTRR~\eqref{eq:general_TRR}
have real coefficients.

As examples we consider the unitary, symmetric matrices
\begin{align}\label{eq:M-matrices}
M^{(1)} = \begin{pmatrix} 1&0\\ 0&1\end{pmatrix}, \quad  M^{(2)} = \begin{pmatrix} 0&1\\ 1&0\end{pmatrix}, \quad M^{(3)} & = \tfrac{1}{\sqrt{2}} \begin{pmatrix} 1&1\\ 1&-1\end{pmatrix}.
\end{align}


By recalling~\eqref{eq:lambda2ist0} and~\eqref{eq:2d_offdiag_laguerre}, the
polynomial $q_{4,6}^{M^{(1)}}$ corresponds to a simple tensor product of
one dimensional Hermite polynomials, while
\begin{equation}\label{eq:M2formula}
q_{7,6}^{M^{(2)}}  = 6! 2^{7} x_1 L^{(1)}_{6}(2x_1x_2).
\end{equation}
One can see the consequences of these simple formulas in the structure
of the nodal sets depicted in the upper panels of figure~\ref{fig:m123}.

The matrix $M^{(3)}$ gives rise to a more complicated
mixing between the two variables. This can also be seen from the
illustration of $q_{6,5}^{M^{(3)}}$ in the lower  panel of figure~\ref{fig:m123}.  It is striking that already
for real matrices $M$  in two dimensions the polynomials
generated by the simple TTRR~\eqref{eq:general_TRR} develop such
nontrivial nodal sets.

\begin{figure}[h!]
\includegraphics[width = 6cm]{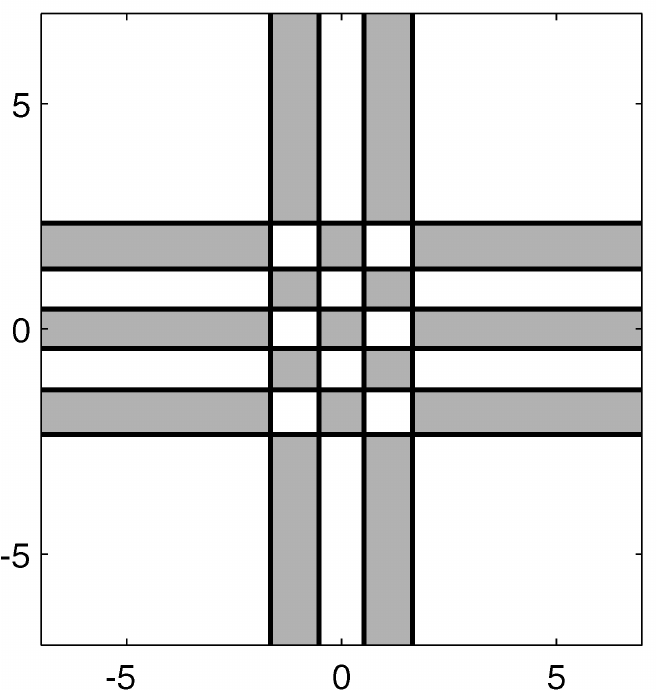}
\includegraphics[width = 6cm]{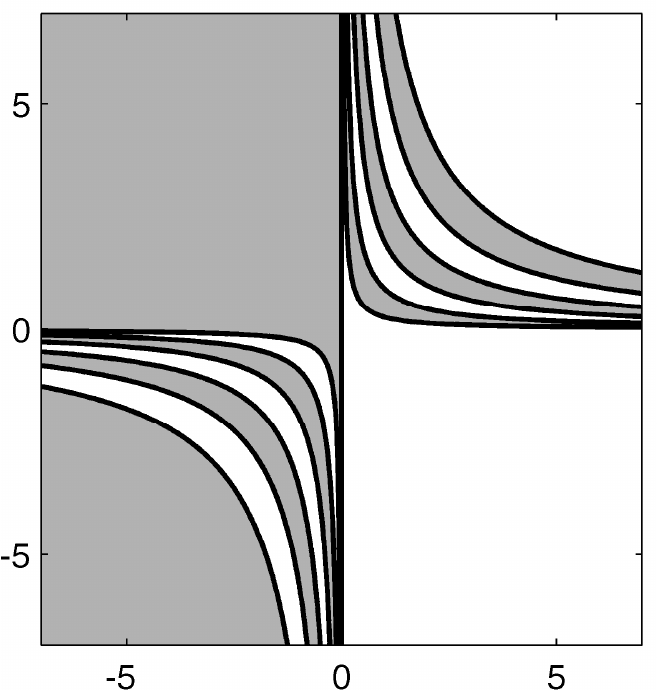}
\includegraphics[width = 6cm]{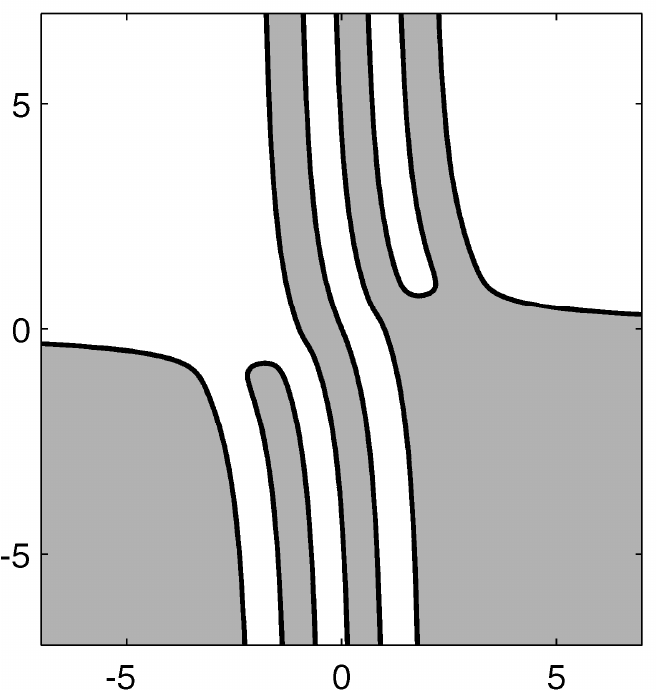}
\caption{The nodal sets of three examplary two-dimensional polynomials,
associated with the matrices $M^{(1)}$ (upper left),
$M^{(2)}$ (upper right) and $M^{(3)}$ (lower).
Regions with negative values are highlighted by grey coloring.}
\label{fig:m123}
\end{figure}
\subsection{Hagedorn wave packets}

In this section we present two-dimensional examples of Hagedorn
wave packets in order to indicate the variety of structures that can be realised.

For our illustrations we employ the same matrices $M^{(j)}$, $j\in\{1,2,3\}$ from~\eqref{eq:M-matrices} as used
for the polynomials in~\S\ref{sec:examples_pol}.
Our choice for the  Lagrangian frames $Z_j = (Q_j;P_j)$ satisfying
\[
Q_j^{-1}\overline{Q_j} = M^{(j)}, \quad j=1,2,3,
\]
 is given by
\begin{align*}
Z_1  = \frac{1}{\sqrt{2}} \begin{pmatrix}1&1\\1 &-1\\ i &i \\ i  &-i \end{pmatrix} ,\quad
Z_2   = \frac{1}{2} \begin{pmatrix}1+i&1-i\\1-i &1+i\\ \frac{i-1}2 & \frac{i+1}2 ,\\
 \frac{i+1}2  & \frac{i-1}2 \end{pmatrix}, \quad Z_3 & = \begin{pmatrix}i&-i(1+\sqrt{2})\\1 &\sqrt{2}-1\\
 \tfrac{1-\sqrt{2}}{2\sqrt{2}} & \tfrac{1}{2\sqrt{2}}  \\\tfrac{i+i\sqrt{2}}{2\sqrt{2}}  &\tfrac{i}{2\sqrt{2}} \end{pmatrix}.
\end{align*}
One can easily check that $Z_1,Z_2$, and $Z_3$ are normalised Lagrangian frames.

We also consider an example of a generalised wave packet associated with the two Lagrangian frames $Z_2$ and $Z_3$.
In this case, the mixing matrix of the polynomials is given by the formula in Proposition~\ref{prop:polynomial_hagedorn_trafo}
and does not equal $Q^{-1}\overline{Q}$.

We stress that despite the fact that the polynomials
\begin{equation*}
q_k^{M_j}, \quad j=1,2,3,
\end{equation*}
 have real coefficients, the wave packets itself are not real-valued since the
polynomials are evaluated on the subspace $Q_j^{-1}\Rd \subset \C^d$, which does
not coincide with $\Rd$ except for $Z_1$.
This case is very special, since one has
\begin{equation*}
 M = \Id \Longleftrightarrow  Q\in \R^{d\times d},
\end{equation*}
which implies that the standard tensor Hermite polynomials $ q_k^{\Id}$ appear only
together with real transformations $Q$. Hence, all Hagedorn wave packets
with $Q^{-1}\overline{Q} = \Id$ correspond to rescaled, sheared,
or shifted multivariate Hermite functions, while this is not true in the case $Q\notin \R^{d\times d}$.

\begin{figure}[h!]
\includegraphics[width = 13cm]{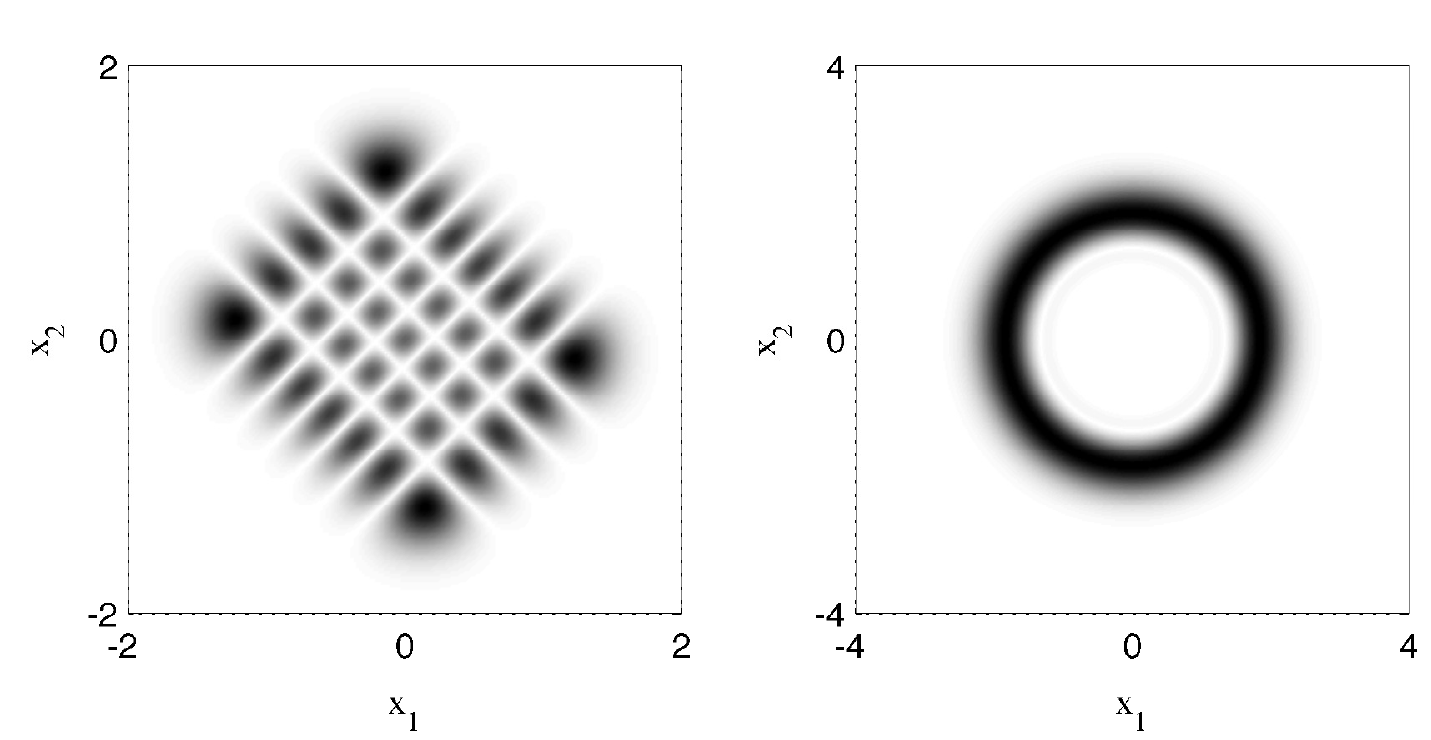}
\caption{
Intensity plot of the absolute value of the two examplary
  two-dimensional Hagedorn wave packets
$\phi_{4,6}^{0.1}[Z_1](x)$ (left) and $\phi_{7,6}^{0.1}[Z_2](x)$ (right), $\eps = 10^{-1}$.
Darker colouring represents higher absolute values.}
\label{fig:hagedorn1und2}
\end{figure}

Figure~\ref{fig:hagedorn1und2} displays the absolute value of two
Hagedorn wave packets associated with the Lagrangian frames $Z_1$ and
$Z_2$. One can recognise that the Hagedorn wave packet
$\phi_{4,6}^{0.1}[Z_1]$ on the left hand side is just a rotated and rescaled
Hermite function, as expected. In contrast, the wave packet
$\phi_{7,6}^{0.1}[Z_2]$ associated with $M^{(2)}$ has a circular
structure, which arises due to a complex rotation of the hyperbolas
from the upper right panel of figure~\ref{fig:m123}.

For  $Z_3$  the resulting wave packets exhibit  complicated
structures. The same is true for generalized wave packets associated with the two
different Lagrangian frames $Z_2$ and $Z_3$,
 as illustrated by the examples in figure~\ref{fig:hagedorn3und4}.

For us, the variety of different Hagedorn wave packets is very fascinating.
We suggest that the selective use of classes of Hagedorn wave packets with specific geometries
could prove useful for designing meshfree numerical discretizations of evolution equations
with a priori known symmetries.


\begin{figure}[h!]
\includegraphics[width = 13cm]{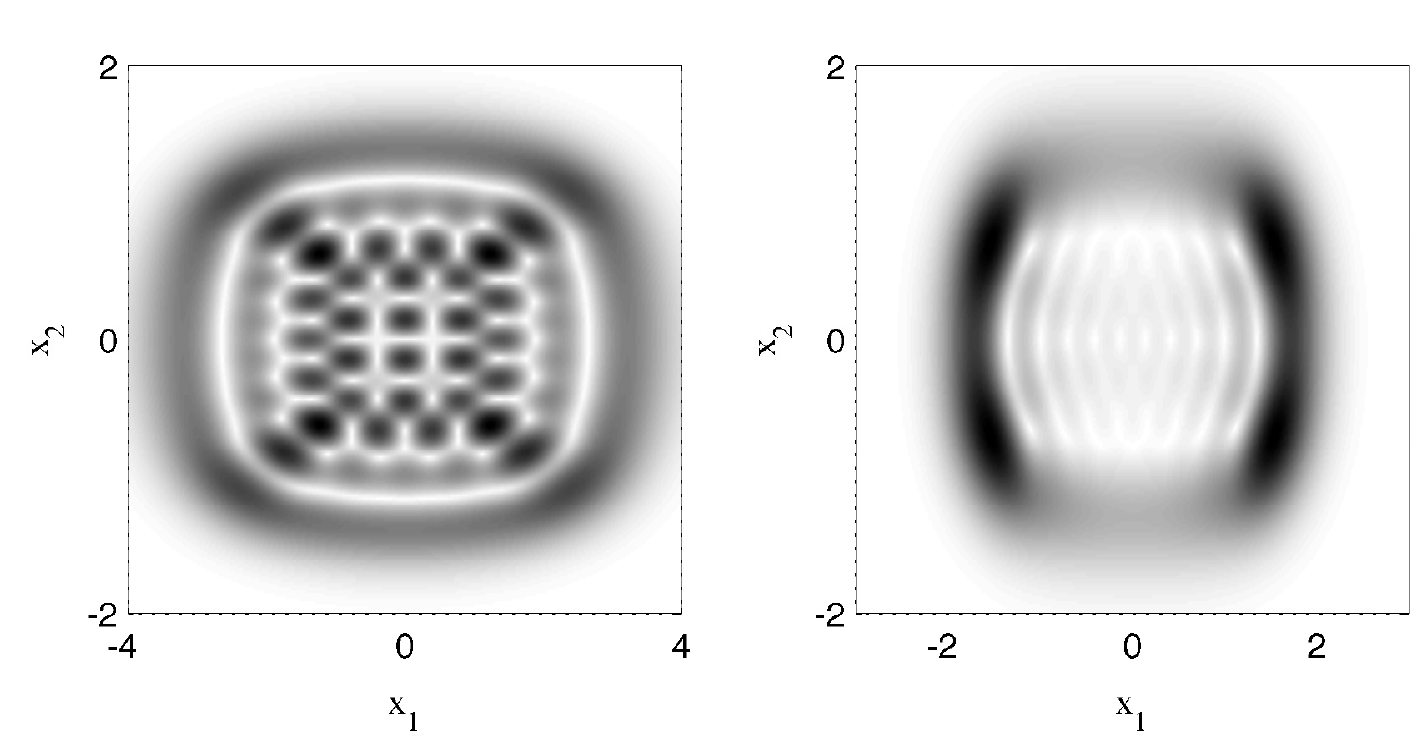}
\caption{
Intensity plot of the absolute value of the
  two-dimensional Hagedorn wave packets
$\phi_{6,5}^{\eps}[Z_3](x)$ (left) and  $\phi_{3,7}^{\eps}[Z_2,Z_3](x)$ (right), where $\eps = 10^{-1}$.}
\label{fig:hagedorn3und4}
\end{figure}

\section{Acknowledgments}
The authors wish to thank  George Hagedorn and Caroline Lasser for
many helpful discussions and valuable comments on the manuscript.

The first named author was supported by the UK Engineering and
Physical Sciences Research Council (EPSRC) grant EP/H023348/1 for the
University of Cambridge Centre for Doctoral Training, the Cambridge
Centre for Analysis. The second and third named authors were supported by the German Research Foundation (DFG),
Collaborative Research Center SFB-TRR 109. The third author gratefully acknowledges support by
the graduate program TopMath of the Elite Network of Bavaria.

\appendix
\section{Proof of Proposition \ref{prop:polynomial_hagedorn_trafo}}\label{app:Proof}

Using the found structure of the polynomials $q_k$
and the corresponding ladder operators, we can give a simple proof of Proposition  \ref{prop:polynomial_hagedorn_trafo}.
\begin{proof}
We prove the assertion by induction over $|k|$. We know that
the raising operator $A^\dagger$ from \eqref{eq:ladder_pos} only creates polynomial prefactors
in front of $\phi_0^\eps$.
The idea is to rewrite $A^\dagger$ as an operator
 acting  on the prefactors only, and identifying it with the raising operator of the polynomials.

Since the case $k=0$ is trivial, we assume the assertion to be true for some $k\in \N^d$.
  Then, for $j=1,\dots,d$, by
  \eqref{eq:hagedorn_raising} we have
  \begin{align*}
    &\phi^\eps_{k+e_j}[Z,Y](x)
    =\frac{1}{\sqrt{k_j+1}}  A_j^\dagger[Y]   \phi^\eps_{k}[Z,Y](x) \\
    &= \frac{2^{-\frac{|k|+1}2}}{\sqrt{(k+e_j)!}}
      \left( \frac{i}{\sqrt{\eps}}
      \left[  X^*(i \eps \nabla_x - PQ^{-1}x) + K^* x \right]_j
      q_k^M \( \tfrac{1}{\sqrt{\eps}} B^* Q^{-1}x \)
      \right)
      \phi^\eps_0[Z](x).
  \end{align*}
  Hence, the result is true as long as
  \begin{equation*}
    q^M_{k+e_j}\( \tfrac{1}{\sqrt{\eps}} B^* Q^{-1}x \)
    = [-\sqrt{\eps} X^* \nabla_x + \tfrac{2}{\sqrt\eps} B^* Q^{-1} x]_j
    q^M_{k}\( \tfrac{1}{\sqrt{\eps}} B^* Q^{-1}x \),
  \end{equation*}
  which follows by invoking the polynomial raising operator of
  Lemma~\ref{lem:raising_operator}. It is easy to see that $M$ is symmetric since the components of $A^{\dagger}$
  commute by the isotropy condition.

The structure of the matrix $M$ follows from the fact that both $Q^{-1}\overline{Q}$ and $QQ^*$ are symmetric and
  \begin{equation*}
  B^*Q^* = \frac{i}{2}(K^*Q-X^*P)Q^* =  \frac{i}{2}(K^*QQ-X^*\overline{P}Q^T-2iX^*)
  \end{equation*}
  since $PQ^* = (QP^*+2i\Id)^T$. Hence,
  \begin{align*}
  B^*Q^*Q^{-T}\overline{B} & = \frac{i}{2}(K^*\overline{Q}\overline{B}-X^*\overline{P}\overline{B})+X^*Q^{-T}\overline{B} =  \frac{i}{2}Y^* \Omega \overline{Z}\overline{B}+X^*Q^{-T}\overline{B}\\
  & = \frac{1}{4} Y^* \Omega \overline{Z} Z^T \Omega \overline{Y}+X^*Q^{-T}\overline{B} = -\frac{1}{4} Y^* \overline{G_{Z_0}} \overline{Y}+X^*Q^{-T}\overline{B},
  \end{align*}
where the last equality is due to the isotropy of $Y$.
\end{proof}

\section{Tensor product representation}\label{app:Tensor}

By recursively applying Proposition~\ref{lem:laguerre+reduction}, one
can derive an expansion of the general polynomials $q_k^M$ in terms of
tensor products of univariate Hermite polynomials
associated with the diagonal entries of $M$.
Moreover, the required
number of summands depends only on the number of offdiagonal entries of $M$ for which  $M_{ij}\neq 0$,
and the corresponding indices $k_i$, and~$k_j$.
\begin{proposition}[Tensor product representation]\label{thm:laguerre}
Let $M\in \C^{d\times d}$ be symmetric, and suppose that there are exactly $n\leq d(d-1)/2$
different off-diagonal index pairs $1\leq \alpha_j <\beta_j \leq d$, $j=1,\hdots,n$,
for which $M_{ \alpha_j \beta_j} = \lambda_j \neq 0$.
Then, for $k\in \N^d$,
\begin{equation}
q_k^M(x) =  \sum_{\substack{\ell \in\N^n \\ \ell_j \leq \min\{k_{\alpha_j},k_{\beta_j}\} }} (-2\lambda)^{\ell}
\ell! {{k_\alpha}\choose \ell}  {{k_\beta}\choose \ell}
\prod_{i =1}^d  H^{M_{i,i}}_{k_i - (E\ell)_i}(x_i)
\end{equation}
with standard multiindex notation, e.g.  $k_\alpha \in \N^n$ with $(k_\alpha)_j = k_{\alpha_j}$.
The index matrix $E\in \N^{d\times n}$ is defined by
\begin{equation}
 E_{ij} = (e_{\alpha_i} + e_{\beta_i})_j.
\end{equation}
\end{proposition}

\begin{proof}
  We start by recalling~\eqref{eq:ladder_matrix_delete}, which can be rewritten as
  \begin{equation}\label{eq:ladder_matrix_delete2}
 q_k^{M}(x) =  \sum_{m=0}^{\min\{k_{\alpha_j},k_{\beta_j}\}} m!
 {{k_{\alpha_j}}\choose m}  {{k_{\beta_j}}\choose m}    (-2 \lambda_j)^m
 q_{k - m(e_{\alpha_j} + e_{\beta_j})}^{M[\alpha_j,\beta_j]}(x) 1,
  \end{equation}
  for all $j=1,\dots,n$.
  One can use the matrix $E$ in order to write
  \begin{equation*}
    k - m(e_{\alpha_j} + e_{\beta_j}) = k - (Em\widehat e_j)
  \end{equation*}
  where $\widehat e_j$ denotes the $j$-th unit vector in
  $\R^n$. Iterating this procedure until all
    offdiagonal entries of $M$ are deleted, completes the proof.
\end{proof}

%
%

%
%
%
%
%
%
%
%
%

\bibliographystyle{abbrv}
\bibliography{biblio}

\end{document}